\documentclass[2col]{article}
\usepackage{misc}
\usepackage{someCustomMacra}
\usepackage{xspace}
\usepackage{amsfonts}
\usepackage{amsmath}
\usepackage{amssymb}
\usepackage{amsthm}
\usepackage{microtype}
\usepackage{tikz}
\usepackage{textcomp}
\usepackage{stmaryrd}
\usepackage{wasysym}
\usepackage{soul}
\usepackage{indentfirst}
\usepackage{enumitem}
\usepackage[T1]{fontenc}
\usepackage[utf8]{inputenc}
\usepackage{lmodern}
\usepackage[hidelinks]{hyperref}

\usetikzlibrary{arrows, arrows.meta, automata, matrix, calc, positioning, shapes.geometric, decorations.pathreplacing}

\usepackage{listings}
\usepackage{graphbox}
\DeclareRobustCommand{\DLLogo}{%
  \begingroup\normalfont
  \kern-1.75pt\includegraphics[align=c,height=1.25\baselineskip]{dl}\kern-1.5pt%
  \endgroup
}

\usepackage{amsthm}
\newtheorem{theorem}{Theorem}
\newtheorem{definition}{Definition}
\newtheorem{example}{Example}
\newtheorem{proposition}{Proposition}
\newtheorem{lemma}{Lemma}

\begin{document}

% \copyrightyear{2024}
% \copyrightclause{Copyright for this paper by its authors.
%   Use permitted under Creative Commons License Attribution 4.0
%   International (CC BY 4.0).}

%   \conference{}
% \conference{\DLLogo{} DL 2024: 37th International Workshop on Description Logics,
%   June 18--21, 2024, Bergen, Norway}

\title{Modal Separability of Fixpoint Formulae}

\author{Jean Christoph Jung \and Jędrzej Kołodziejski}

\date{TU Dortmund University}
% orcid=0000-0002-4159-2255,
% email=jean.jung@tu-dortmund.de

% \author{Jędrzej Kołodziejski}%
% orcid=0000-0001-5008-9224,
% email=jedrzej.kolodziejski@tu-dortmund.de,
% url=https://www.mimuw.edu.pl/~jkol/,

% \address[1]{TU Dortmund University, Germany}
% \address[2]{TU Dortmund University, Germany}

\maketitle
\begin{abstract}
  We study \emph{modal separability} for fixpoint formulae: given two
  mutually exclusive fixpoint formulae $\varphi,\varphi'$, decide whether there is a modal
  formula $\psi$ that separates them, that is, that satisfies
  $\varphi\models\psi\models\neg\varphi'$. This problem has
  applications for finding simple reasons for inconsistency. Our main
  contributions are tight complexity bounds for deciding modal
  separability and optimal ways to compute a separator if it exists.
  More precisely, it is \ExpTime-complete in general and
  \PSpace-complete over words. Separators can be computed in doubly
  exponential time in general and in exponential time over words, and
  this is optimal as well. The results for general structures transfer
  to arbitrary, finitely branching, and finite trees. The word case results
  hold for finite, infinite, and arbitrary words.
\end{abstract}

% \begin{keywords}
%   Modal Logic\sep
%   Fixpoint Logic\sep
%   Separability\sep
%   Interpolation
% \end{keywords}


\section{Introduction}

For given logics $\Lmc,\Lmc^+$, the \emph{$\Lmc$-separability problem
for $\Lmc^+$} is to decide given two $\Lmc^+$-formulae $\varphi,\varphi'$
whether there is an \Lmc-formula $\psi$ that separates
$\varphi$ and $\varphi'$ in the sense that $\varphi\models\psi$ and
$\psi\models\neg\phi'$. Obviously, a separator can only exist when
$\varphi$ and $\varphi'$ are mutually exclusive, and the problem is only
meaningful when $\Lmc$ is less expressive than $\Lmc^+$.
Intuitively, a separator formulated in a ``simpler'' logic \Lmc explains a given inconsistency in a ``complicated'' logic $\Lmc^+$.
Note that, for logics $\Lmc^+$ closed under negation, \Lmc-separability
generalizes the \emph{\Lmc-definability problem for $\Lmc^+$}: decide
whether a given $\Lmc^+$-formula is equivalent to an $\Lmc$-formula.
Indeed, $\varphi\in \Lmc^+$ is equivalent to an \Lmc-formula iff
$\varphi$ and $\neg\varphi$ are \Lmc-separable. Since separability is more
general than definability, solving it requires an even
better understanding of the expressive power of the logics under
consideration.

\begin{example} \label{ex:alcreg}
Consider $\Lmc$ being the modal logic $\ML$, also known under the name
\ALC in the context of description logics. Expressions of the logic
(called \emph{formulae} in $\ML$ terminology and \emph{concepts} in
description logic parlance) describe properties of colored, directed
graphs with a distinguished point called the root. As the more
expressive $\Lmc^+$ take $\PDL$: the extension of $\ML$ with regular
modalities (in DL terms: the extension $\ALC_\textit{reg}$ of \ALC with regular role expressions). Assume the graphs under consideration have edges labelled with colors $A$, $B$ and $C$ and consider properties:
\begin{itemize}
  \item[$P$:] ``There is a path from the root whose labeling belongs to $A^+B$.''
  \item[$P'$:] ``The labeling of every (finite) path from the root belongs to $C^*$.''
\end{itemize}
These (contradictory) properties are expressed by $\PDL$-formulae
$\phi$ and $\phi'$ and it is easy to see that none of them can be
expressed in the weaker $\ML$. Nonetheless, $\phi$ and $\phi'$ are
separated by a simple $\ML$-formula $\psi$ that says: ``there is an $A$-labelled edge from the root''. Thus, $\psi$ serves as an easy explanation of the inconsistency of $\phi$ and $\phi'$.

\end{example}
Generalizing the example, in this paper we investigate
$\ML$-separability of formulae in the
modal $\mu$-calculus
$\muML$~\cite{DBLP:conf/focs/Pratt81,DBLP:journals/tcs/Kozen83}, which
extends $\PDL$~\cite{DBLP:journals/jcss/FischerL79}.
$\muML$ is a general framework capturing logics supporting fixpoints
that is
relevant both for knowledge representation and reasoning and for
verification. It describes all bisimulation-invariant properties
definable in $\MSO$~\cite[Theorem 11]{JW96} and thus encompasses
virtually all specification languages such as $\LTL$ and
$\CTL$~\cite{DBLP:books/daglib/0020348}.

Our results generalize the $\ML$-definability problem for $\muML$ which was shown decidable by Otto~\cite[Main Theorem]{DBLP:conf/stacs/Otto99}. The adaptation of the argument to the more general separability is relatively easy. However, Otto's paper is focussed on deciding the existence of modal definitions. The problem of computing a definition when it exists is not discussed, and it seems that the formula which can be read off from the proof is at least tower-exponential big. This issue was addressed in~\cite{DBLP:conf/csl/LehtinenQ15}. Unfortunately, in this case the approach, although constructive, does not easily generalize from definability to separability. Matching lower bounds are also missing.
We fill the gap by providing a procedure which is \emph{both} fully constructive and works for the more general separation case. Both the constructed formulae and the running time are optimal, as illustrated by suitable examples and reductions.
% Our procedures are modular and split into two independent tasks: computing appropriate bounds and using them to construct a formula. For the first task we rely on the ideas of Otto. The latter one is new.

We consider both general models and ``word models'' which are Kripke
structures in which each point has at most one successor. The latter
are relevant from a verification perspective and for temporal
reasoning. In order to obtain our results we first prove
model-theoretic characterizations in terms of bisimulations. We then
exploit the close connection of $\muML$ to nondeterministic parity
tree automata to give (1)~optimal
procedures for the separability problem and (2)~upper bounds on the
modal depth of a separator, if it exists. In~(1) we show
\ExpTime-completeness of modal separability in general and
\PSpace-completeness over words. The lower bounds are essentially
inherited from satisfiability.
The upper bounds derived in~(2) are
then used together with the automata to compute so-called 
\emph{$n$-uniform consequences}, that is, modal formulae that have
exactly the same modal consequences as a given $\muML$-formula, up to
modal depth $n$. 
% Thus, $n$-uniform consequences are reminiscent of
% uniform interpolants. 
These $n$-uniform consequences are then used as 
separators. Also here, our procedures are optimal: they compute
separators of at most double exponential size, and we show that there
are $\muML$-formulae that are expressible in $\ML$ but any equivalent
$\ML$-formulae must have doubly exponential size. This means that
there is a double exponential succinctness gap between $\muML$ and
$\ML$. In the word case, our procedures compute exponentially sized
separators and there is only an exponential succinctness gap. All
lower bounds (both computational and succinctness) already hold for
$\PDL$ ($\ALC_{\textit{reg}}$) in place of $\muML$, and for
definability in place of separability.

It is interesting to note that our results hold over classes of models
definable by $\muML$-formulae. This observation allows us to cover the
more general notion of separation in presence of an \emph{ontology}
(i.e.~a background theory imposing some conditions on models). As long
as the ontology is expressible in $\muML$, separability and
computation of separators reduce to the ontology-free setting.
Without much effort the same observation lets us transfer our results
to finite words, infinite words, and finite trees.

Missing proofs can be found in the appendix.

\paragraph{Related Work.}
Separability has been intensively studied in formal language theory. A seminal result in this area is that separability of regular word languages by a first-order language is decidable in \ExpTime~\cite{DBLP:journals/corr/PlaceZ14}. 
Since $\muML$ over words defines precisely the regular languages and first-order logic captures $\ML$, this is particularly related to our results over words. 

In logic, a recent work investigates the complexity of separating
between formulae supporting counting quantifiers by formulae that do
not support them~\cite{separatingcounting}. The used techniques exploit compactness, which makes them inapplicable to our case and inherently non-constructive.

Another related problem is the question of interpolant existence. An \emph{interpolant} of two \Lmc-formulae $\varphi$ and $\varphi'$ is an
\Lmc-formula $\psi$ with $\varphi\models\psi\models\varphi'$ and such
that the signature of $\psi$ is contained in the signatures of both
$\varphi$ and $\varphi'$. Thus, the problem resembles separability but
the restriction on $\psi$ is in terms of the signature instead of in
terms of the logic. Sometimes this question reduces to entailment, as
many logics enjoy the \emph{Craig interpolation property}: an
interpolant of $\phi$ and $\phi'$ exists whenever $\phi\models\phi'$.
Interpolant existence for logics that lack Craig interpolation has
recently been studied in~\cite{DBLP:conf/lics/JungW21,
DBLP:journals/tocl/ArtaleJMOW23}. The used tools, however, are similar
in nature to the ones from~\cite{separatingcounting} and therefore
inapplicable to our problem.

Finally, a related problem is \emph{separability of data examples}.
There, the task is to separate sets of
pointed structures instead of formulae (see~\cite{DBLP:journals/is/Martins19,DBLP:journals/ai/JungLPW22} and
the references therein). Separability of data examples can be cast as an
instance of (our logical notion of) separability if $\Lmc^+$ is expressive
enough to describe the data examples. Conversely, \Lmc-separability
of formulae $\phi$ and $\phi'$ is the same as data separability of the
(possibly infinite) sets of their models by an \Lmc-formula.

\section{Preliminaries}

Assuming familiarity of the reader with modal logic and the
modal $\mu$-calculus, we recall here only the main notions and refer
to~\cite{DBLP:books/el/07/BBW2007} for more details.

\paragraph{Syntax.}
% \noindent{\bf Syntax.}
% Assume a \emph{modal signature} consisting of two finite sets: \emph{actions} $\Actions$ and \emph{propositions} $\Propositions$. 
We consider modal logic $\ML$ and its fixpoint extension $\muML$ over
a \emph{modal signature} consisting of two finite sets: \emph{actions} $\Actions$ and \emph{propositions} $\Propositions$. The syntax of $\ML$ is given as:
\[
  \phi ::= \top\ |\ \bot\ |\ \atProp\ |\ \neg\atProp\ |\ \phi\vee\phi\ |\ \phi\wedge\phi\ |\ \diamond{\action}\phi\ |\ \boxmodal{\action}\phi\
\]
with $\atProp\in\Propositions$ and $\action\in\Actions$. If
$\Actions=\{\action\}$ is a singleton, we use $\Diamond\phi$ and $\Box\phi$ in place of
$\diamond{\action}\phi$ and $\boxmodal{\action}\phi$. The syntax of $\muML$ is obtained by extending the above with additional clauses:
\[
  \phi ::= x\ |\ \mu x.\phi\ |\ \nu x.\phi
\]
where $x$ belongs to a fixed set $\Var$ of variables. The restriction
to a fixed finite signature is only for the sake of readability. All results
in the paper remain true with arbitrary signature.

\paragraph{Semantics.}
% \medskip \noindent{\bf Semantics.} 
The models we consider are pointed
Kripke structures. That is, a model $\M$ consists of a set $M$ (called
its \emph{universe}) with a distinguished point $\point_I\in M$ called
the \emph{root}, an interpretation
${\arrowAction}\subset M\times M$ for every
$\action\in\Actions$ and a valuation
$\val:M\to\powerset{\Propositions}$. We call the set
$\powerset{\Propositions}$ \emph{colors} and denote it by $\alphabet$.
Both $\ML$ and $\muML$ are interpreted in points of models in a
standard way. Since models are by definition
pointed we write $\M\models\phi$ meaning that the root of $\M$
satisfies $\phi$. 
% We call two models $\M$ and $\N$ \emph{modally
% indistinguishable} and wite $\M\equiv_{\ML}\N$ if they satisfy the
% same $\ML$ formulae.  
The same symbol denotes entailment:
$\phi\models\psi$ means that every model of $\phi$ is a model of
$\psi$. In the case only models from some fixed class $\C$ 
are considered we talk about satisfiability and entailment
\emph{over} $\C$ and in the latter case write
$\phi\models_\C\psi$.

A particularly relevant class of models are trees. 
% In general, infinite and infinitely branching models are allowed, but
% we consider other classes of models. 
A model $\M$ is a \emph{tree} if the underlying directed
graph $(M,\bigcup\{\arrowAction\ |\ \action\in\Actions\})$ is a tree
with $v_I$ as its root. The \emph{branching} or \emph{outdegree} of a
point is the number of its children in this underlying graph.
% that is, we count the branching over all the actions together. 
The class of
all trees is denoted by $\Trees{}$. 
%, the class of trees of outdegree at most $k$ by $\Trees{\leq k}$ and the class of finitely branching trees by $\Trees{<\omega}$.
% \emph{Words} are trees of outdegree at most one using a single fixed action.
We identify \emph{words} (both finite and infinite) over alphabet
$\Sigma$ with trees over a single action of outdegree at most one.
Points of such models are interpreted as positions in the word, the
unique accessibility relation represents the successor relation, and
the valuation determines the letter at each position.
A \emph{prefix} of a tree is a subset of its universe closed under
taking ancestors. When no confusion arises we identify a prefix
$N\subset M$ with the induced subtree $\N$ of $\M$ that has $N$ as its
universe. The \emph{depth} of a point is the distance from the root.
The prefix of depth $n$ (or just \emph{$n$-prefix}) is the set of all
points at depth at most $n$ and is denoted by $M_{|_n}$ (and the
corresponding subtree by $\M_{|_n}$).
% The \emph{tallness} of a tree is the distance from the root to the \emph{closest} leaf.\nb{J: not used in main part}

We define bisimulations and bisimilarity in a standard way except that
in the case of trees for convenience we assume that bisimulations only
link points at the same depth. An $n$-step bisimulation (or just
\emph{$n$-bisimulation}) between trees $\M$ and $\N$ is a bisimulation
between their $n$-prefixes. We denote 
%  bisimilarity by $\bis$\nb{J: not used} and
$n$-bisimilarity by $\bis^n$.

% For simplicity we usually only describe constructions in the monomodal
% case (i.e.~with $|\Actions|=1$), skip the unique label and write
% $\arrowActionLabel{}$, $\Diamond$ and $\Box$.\nb{here?}

% \medskip \noindent{\bf Size of formulae.} 
\paragraph{Size of formulae.} 
The \emph{size} of a formula $\varphi$, denoted $|\varphi|$, is the number of nodes in its syntax
tree. Similarly, its \emph{depth} is the maximal length of paths in the syntax tree. The
depth of a formula should not be confused with its \emph{modal depth}
which is the maximal nesting of modal operators; all formulae of modal depth at most $n$ are denoted $\ML^n$.

When we specify formulae in the paper, we use syntactic sugar
$\bigvee\Phi$, $\bigwedge\Phi$, and \emph{nabla}
$\nabla\Phi$, for finite sets of formulae $\Phi$.
The first two are self-explanatory and allow for higher branching in
the syntax tree. The last one, $\nabla\Phi$,
intuitively means that ``every formula in $\Phi$ is true in some child
and every child satisfies some formula from $\Phi$'' and is an
abbreviation for 
\begin{align}\label{eq:nabla}
  \textstyle\nabla\Phi = \bigwedge_{\phi\in\Phi}\Diamond\phi \wedge \Box\bigvee_{\phi\in\Phi}\phi.
\end{align}
It is well-known that $\bigvee\Phi$ and $\bigwedge\Phi$ can be
rewritten into basic syntax under polynomial cost. We also include the colors $\alphabet$
directly in the syntax: $c\in\alphabet$ is a shorthand for the formula
$\bigwedge\{\atProp,\neg\atProp'\ |\ \atProp\in c, \atProp'\notin
c\}$. Rewriting colors increases the size only by a factor
linear in $|\Propositions|$.
% 
% Note that rewriting colors $\alphabet$ to the basic syntax only
% increases the size of formulae by a factor linear in the number of
% propositions. Moreover, in a formula of size $k$ each $\bigvee$ can be
% replaced with $\log k$ binary $\vee$, and similarly with $\bigwedge$
% and $\wedge$. Thus, all $\bigvee$ and $\bigwedge$ can be eliminated
% only at a polynomial cost. With $\nabla$ one has to be more careful,
% as rewriting the formula using equation~\eqref{eq:nabla} may blow up
% its size exponentially. We discuss this issue in more detail in
% \textcolor{red}{TODO!!!} The question whether adding $\nabla$ to $\ML$ makes it exponentially more succinct is left open.

\paragraph{\bf Automata.} 
% \medskip \noindent{\bf Automata.} 
Throughout the paper we
use automata over tree models of both bounded and arbitrary outdegree.  A
\emph{nondeterministic parity tree automaton (NPTA)} is a tuple
$\A=(Q,\Sigma,q_I,\delta,\rank)$ where $Q$ is a finite set of states,
$q_I\in Q$ is the initial state, $\Sigma$ is the alphabet fixed above,
and $\rank$ assigns each state a
priority. The transition function $\delta$ is of type:
\[
\delta:Q\times\alphabet \to \powerset{\powerset{Q}}.
\]
Intuitively,
$\delta(q,c)=\{S_1,...,S_l\}$ means that in the state $q$ upon reading
color $c$ the automaton (i) chooses a transition $S_i$ and (ii) labels
all the children of the current point with states from $S_i$ so that
every $p\in S_i$ is assigned to some child. A run of $\A$ on a tree
$\M$ is an assignment $\rho:M\to Q$ consistent with $\delta$ in such
sense and sending the root of the tree to $q_I$. The run is accepting
if for every infinite path $v_0,v_1\ldots$ in $\M$ the sequence
$\rank(\rho(v_0)),\rank(\rho(v_1)),\ldots$ satisfies the parity
condition. We write $\M\models\Amc$ in case \Amc has an accepting run
on $\M$. An automaton that is identical to $\A$ except that the
original initial state is replaced with $q$ is denoted
$\A[q_I\mapsfrom q]$. We refer with NPWA to an NPTA working over
words. 

In NPTAs over trees of bounded outdegree $k$ it might be more common
to use a transition function of type $\delta:Q\times \alphabet\to
\powerset{Q^k}$, but the difference is not essential: our NPTAs can be
represented in this way and conversely, all relevant constructions for
such NPTAs can be adapted to our setting. Most importantly, we rely on
the following classical result (see for example the discussion
in~\cite{DBLP:conf/icalp/Vardi98} and the well-presented Dealternation Theorem~5.7 in~\cite{BojanCzerwinski18}):
\begin{theorem} \label{thm:munpta}
  For every $\muML$-formula $\varphi$, we can construct an equivalent
  NPTA $\Amc$, that is, $\Mmc\models \varphi$ iff $\Mmc\models \Amc$,
  for every tree $\Mmc$, with number of states at most exponential in
  $|\varphi|$. If we consider models of bounded outdegree $k$ then
  $\Amc$ is computed in exponential time, otherwise in doubly exponential time. 
\end{theorem}

\section{Foundations of Separability}\label{sec:separability}

We start with recalling the notion of separability and discuss some
of its basic properties. 
\begin{definition}
  Given $\phi,\phi'\in\muML$, a \emph{modal separator of
  $\varphi,\varphi'$} is $\psi\in\ML$
  with $\phi\models\psi$ and $\psi\models\neg\phi'$. 
  It is a modal
  separator \emph{over a class} $\C$ if $\phi\models_\C\psi$ and
  $\psi\models_\C\neg \varphi'$. 
\end{definition}
The notion induces the problem of \emph{modal separability}: given two
$\muML$-formulae
$\varphi,\varphi'$, decide whether a modal separator exists. Clearly, $\ML$-definability of $\varphi$ or $\varphi'$ is a
sufficient condition for the existence of a modal separator between
$\varphi,\varphi'$. However, Example~\ref{ex:alcreg} shows that it is
not a necessary one: neither $\varphi$ nor $\varphi'$ are 
$\ML$-definable, yet a separator exists. We make some foundational
observations. 

Inspired by the notion of Craig interpolation, one could also consider the notion of a
\emph{Craig} modal separator, which is a modal separator $\psi$ of
$\varphi,\varphi'$ which only uses symbols occurring in both $\phi$ and $\phi'$. However, based on the fact that $\ML$ enjoys
Craig interpolation, we show in Theorem~\ref{thm:craig} (proof in
Appendix~\ref{app:craig}) that Craig modal separability and modal
separability coincide. Since $\ML$ enjoys Craig interpolation over 
many classes of models~\cite[Theorem 1]{marx99}, Theorem~\ref{thm:craig} remains true over all classes
of models considered below. We thus focus on modal separability.
% \nb{J:
%   actually we should focus on Craig modal separability in the
% computation, right? It's stronger} 
% \nb{found this sentence commented: Moreover, by~\cite[Theorem
% 1]{marx99} $\ML$ has the interpolation property over every class
% satisfying some simple conditions. These sufficient conditions are
% satisfied by all the classes of models we consider such as words or
% finite trees and also preserved by $\muML$-definitions in the sense of
% Lemma~\ref{lem:reduction-definable}. }
%
\begin{theorem}\label{thm:craig}
  $\varphi,\varphi'\in\muML$ admit a modal separator iff they admit a
  Craig modal separator. 
\end{theorem}
Inspired by the notion of \emph{uniform interpolation}~\cite{Visser_1996,
DBLP:conf/aiml/DAgostinoH96}, it is natural to ask whether every $\varphi\in\muML$
admits a \emph{uniform modal separator}, that is, a formula $\psi\in \ML$
that is a modal separator of $\varphi,\varphi'$ for every
$\varphi'\in\muML$ with $\varphi\models\neg \varphi'$. However, substituting $\neg\phi$ for $\phi'$ we get that the uniform modal separator $\psi$ for $\phi$ is actually equivalent to $\phi$. Consequently, a $\muML$-formula has a uniform modal separator iff it is modally definable.
This is contrast with
the fact that both $\ML$~\cite{Visser_1996} and $\muML$~\cite{DBLP:conf/aiml/DAgostinoH96} enjoy uniform
interpolation.

Since $\muML$ has both the finite model property and the (finitely
branching) tree model property, the notions of a modal separator over
finite models, arbitrary tree models, and finitely branching tree
models all coincide with modal separator (over arbitrary models).
Unsurprisingly, this does not apply to the class of all finite trees.
\begin{example}\label{ex:finitetrees}
  Consider a $\muML$-formula $\varphi_\infty=\nu x.\Diamond x$ expressing that there
  exists an infinite path originating in the root.
%   Formula $\varphi_\infty$ from Example~\ref{ex:uniform} is
  It is
  satisfiable, but unsatisfiable over finite trees. Thus
  $\bot$ is an $\ML$-definition of $\varphi_\infty$ over finite trees,
  but $\varphi_\infty$ is not $\ML$-definable (over arbitrary models).
\end{example}
%
% To deal with separability over finite trees we use the following
% observation. 
We deal with separability over finite trees as follows.
Call a class $\C$ of models \emph{$\muML$-definable in
$\D$} if there is a $\muML$-formula $\theta$ such that
$\M\in\C$ iff $\M\models\theta$, for all models $\M\in \D$.
\begin{lemma} \label{lem:reduction-definable}
  Let $\C$ be $\muML$-definable in $\D$ by $\theta$ and let
  $\psi\in\ML$. Then $\psi$ is a modal separator of
  $\varphi,\varphi'\in\muML$ over $\C$ iff $\psi$ is a modal separator
  of $\theta\wedge\varphi$ and $\theta\wedge\varphi'$ over
  $\D$.
  % \begin{gather*}
  %   \phi\models_\C\psi \hspace*{0.3cm} \text{and} \hspace*{0.3cm} \psi\models_\C\neg\phi'\\
  %   \iff\\
  %   \phi\wedge\theta\models_\D\psi \hspace*{0.3cm} \text{and} \hspace*{0.3cm} \psi\models_\D\neg\phi'\wedge\theta
  % \end{gather*}  
  %
\end{lemma}
Intuitively, Lemma~\ref{lem:reduction-definable} provides us with a
reduction of modal separability over $\Cmc$ to modal
separability over (the larger) $\Dmc$. It has a number of interesting 
consequences. First, observe that the formula $\neg\varphi_\infty$ from
Example~\ref{ex:finitetrees} defines
the class of finite trees in the class of all finitely branching
trees. Hence $\neg\phi_\infty$ provides a reduction of modal separability over finite trees
to modal separability over finitely branching trees, and thus to modal
separability.
% Second, it shows that modal separability over arbitrary
% $\muML$-definable classes reduces to modal separability, again using
% $\varphi_\infty$. 
Similarly, and again using $\varphi_\infty$, Lemma~\ref{lem:reduction-definable} reduces modal separability over
finite words and over infinite words to modal separability
over (arbitrary) words.
Finally, the lemma
can be used to reduce modal separability relative to background knowledge
to modal separability. Call $\psi$ a modal separator of
$\varphi,\varphi'\in \muML$ \emph{relative to $\theta_0\in\muML$} if
it is a modal separator of $\varphi,\varphi'$ over the class of models
satisfying $\theta_0$ \emph{in every point}. This setting is
most relevant for the DL community since $\theta_0$ plays the
role of an ontology. In particular, the question whether two
$\ALC_\textit{reg}$-concepts $\varphi,\varphi'$ are separable by an
\ALC-concept relative
to an $\ALC_\textit{reg}$-ontology is an instance of that problem
(recall that every $\ALC_\textit{reg}$-concept can be expressed as a
$\muML$-formula). Let $\theta$
be the $\muML$-formula expressing that $\theta_0$ is satisfied in
every point reachable via the accessibility relations. Using
Lemma~\ref{lem:reduction-definable} and bisimulation invariance of $\muML$,
it is routine to verify that $\psi$ is a modal separator of $\varphi,\varphi'$
relative to $\theta_0$ iff $\psi$ is a modal separator of
$\theta\wedge \varphi$ and $\theta\wedge\varphi'$. 

In view of what was said so far, we will from now on concentrate on
deciding modal separability over general and word models and computing
separators if they exist. A main ingredient for both tasks is to show 
that if there is a modal separator for
$\muML$-formula $\varphi,\varphi'$, then there is one of modal depth $n$
at most exponential in the size of $\varphi$ and $\varphi'$. % bounded by FUNCTION is trivial cause there are only finitely many formulae of size k
As a necessary tool for showing this exponential bound on $n$, and for efficiently deciding if a given $n$ suffices, in Appendix~\ref{app:model theory} we establish the following model-theoretic characterization. 
%  for the existence of modal
% separators for a given modal depth $n$.
Fix $\phi,\phi'\in\muML$ for
the rest of the paper and denote their size by $k=|\phi|+|\phi'|$. 

\begin{proposition}\label{prop:separability reformulation}
Let $n\in\NN$. The following are equivalent:
\begin{enumerate}[label=(\roman*)]
  \item\label{it:separability} There is $\psi\in\ML$ of modal depth $n$ separating $\phi$ and $\phi'$;
  \item\label{it:bound bisimilar} For all models $\M$ and $\M'$
    \underline{bisimilar} up to depth $n$: $\M\models\phi$ implies
    $\M'\not\models\phi'$;
  \item\label{it:bound identical} For all trees $\M$ and $\M'$
    \underline{identical} up to depth $n$: $\M\models\phi$ implies
    $\M'\not\models\phi'$;
 \item\label{it:bound id + bounded} For all trees $\M$ and $\M'$ \underline{identical} up to depth $n$ and whose \underline{branching is bounded} by $k$:\\ $\M\models\phi$ implies $\M'\not\models\phi'$.
\end{enumerate}
\end{proposition}
Based on Proposition~\ref{prop:separability reformulation}, we show
that $\ML$-separability of $\muML$-formulae is \ExpTime-complete and
thus not harder than $\ML$-definability. 
\begin{theorem}\label{thm:models-complexity}
  Modal separability of $\muML$-formulae is \ExpTime-complete over
  arbitrary models. 
\end{theorem}
\ExpTime-hardness already holds for $\ML$-definability and is proved by an immediate
reduction from $\muML$-satisfiability, which is \ExpTime-complete
already for its fragment $\PDL$~\cite[Section
4]{DBLP:journals/jcss/FischerL79}.
%  Indeed, let $\varphi_0$ be a
% $\PDL$-formula not equivalent to an $\ML$-formula. Then
% $\varphi\wedge\varphi_0$ is $\ML$-definable iff $\varphi$ is
% unsatisfiable.
% The above reduction does not work! consider phi0='everything is red' and phi='no path longer than 2'. You can reduce using relativization but I would just skip it.
It is not hard to modify the original hardness proof
for $\PDL$-satisfiability to work over finite trees, so
Theorem~\ref{thm:models-complexity} remains valid over finite trees as
well.
For the upper bound, we mostly follow the technical development
in~\cite{DBLP:conf/stacs/Otto99}. Thanks to Proposition~\ref{prop:separability reformulation} separability is equivalent to the existence of $n\in\NN$ for which condition~\ref{it:bound id + bounded} holds. This can be expressed as an $\MSO$ statement about the full $k$-ary tree, and thus decided. However, for optimal complexity and to extract bounds that we use later, in Appendix~\ref{app:deciding and bounds - all models} we apply a lower-level automata-theoretic analysis.
\\

Over words, we essentially follow the same approach. Since the tree
automata used in the proof of Theorem~\ref{thm:models-complexity} can
be replaced by word automata, the complexity drops to \PSpace (see Appendix~\ref{app:deciding and bounds - words}).
A matching lower bound can be derived as above by a
reduction from satisfiability in $\LTL$ over
words~\cite[Theorem~4.1]{DBLP:journals/jacm/SistlaC85} (which, in fact, can be rephrased
in terms of $\PDL$).
\begin{theorem}\label{thm:word-complexity}
  Modal separability of $\muML$-formulae is \PSpace-complete over
  words.
\end{theorem}

As announced, an important step in the proofs of the upper bounds, both in the case with arbitrary models and with words, is the following proposition which we will also use later.
\begin{proposition}\label{prop:bound on nesting}
  If $\phi,\phi'\in\muML$ are separable then they are separable by a formula of modal depth $l$ exponential in their size $k$. The same is true over words.
\end{proposition}

In the remainder of the paper we will deal with computing separators
based on Proposition~\ref{prop:bound on nesting}. Before we proceed let note that our approach differs from the treatment of modal definability from~\cite{DBLP:conf/csl/LehtinenQ15}. There, the authors rewrite given $\phi$ into modal $\psi$ in such a way that if the initial $\phi$ is modally definable then $\phi$ and $\psi$ are equivalent. In the case when $\phi$ is not modally definable, however, the output $\psi$ is rather random. For example, $\psi$ obtained from the formula $\phi_\infty$ from Example~\ref{ex:finitetrees} is equivalent to $\bot$ which is not even a consequence of $\phi_\infty$. Thus, a different construction is needed to obtain separators. We will actually compute something slightly stronger that might be of independent
interest. 
% 
% 
% 
% Let us introduce a notion that proves useful in the context of
% separability.
% 
\begin{definition}\label{def:n-uniform consequence}
  Given $\phi\in\muML$ and $n\in\NN$, a formula $\psi\in\ML^n$ is an
  \emph{$n$-uniform consequence} of $\phi$ if, for all $\theta\in
  \ML^n$:
  \[ \phi\models\theta \hspace*{0.5cm} \iff \hspace*{0.5cm}
  \psi\models\theta \]
  An analogous notion relative to a fixed class $\C$ of models is
  obtained by replacing $\models$ with $\models_\C$.
\end{definition}
In words: $\psi$ is an $n$-uniform consequence of $\phi$ if it has
modal depth $n$, is a consequence of $\phi$, and entails every other
consequence of $\phi$ of modal depth $n$. In particular, if $\phi$ and
$\phi'$ are separable by \emph{some} modal formula of modal depth $n$
and $\psi$ is an $n$-uniform consequence of $\phi$, then this $\psi$
separates $\phi$ from $\phi'$ as well. Observe that
$n$-uniform consequences exist for every $\varphi\in\muML$ and
$n\in\NN$. Indeed, given $\phi$ and $n$ we can obtain an
$n$-uniform consequence $\psi$ of $\varphi$ by taking the disjunction of all
$\ML^n$-types consistent with $\phi$. Here, by an \emph{$\ML^n$-type}
we mean a maximal consistent subset of $\ML^n$. Since up to
equivalence there are only finitely many formulae in $\ML^n$, each
$\ML^n$-type can be represented as a single $\ML^n$-formula and the
mentioned disjunction $\psi$ is well-defined.

In view of Proposition~\ref{prop:bound on nesting}, it thus suffices to compute $n$-uniform consequences of
$\varphi$. Unfortunately, the naive construction given above is
nonelementary in the size of the separated formulae $\phi$ and
$\phi'$. In the next sections we give better constructions.

\section{Optimal Separators: Arbitrary Models}\label{sec:optimal separators}

 We construct doubly exponentially sized separators and provide matching lower bounds.

\subsection{Construction}

\begin{theorem}\label{thm:construction}
If $\phi$ and $\phi'$ are modally separable then a separator $\phi$ of size doubly exponential in $k=|\phi|+|\phi'|$ exists and can be computed in doubly exponential time.
\end{theorem}
The above is a consequence of the following lemma.

\begin{lemma}\label{lem:construction}
For every $\phi\in\muML$ and $n\in\NN$, one can construct an $n$-uniform consequence $\psi_n\in\ML^n$ of $\phi$ with branching doubly exponential in $|\phi|$ and depth linear in $n$.
\end{lemma}
We show how Theorem~\ref{thm:construction} follows from
Lemma~\ref{lem:construction}. Proposition~\ref{prop:bound on nesting}
guarantees that if a modal separator for $\phi$ and $\phi'$ exists
then there is one with modal depth $l$ exponential in $k$. Since
$\psi_l$ entails this separator it follows that $\psi_l$ is a
separator itself.

The branching $m$ of $\psi_l$ is at most doubly exponential in $|\phi|$ and thus also in $k$: $m\leq 2^{2^{k^x}}$ for some constant $x$. The depth $d$ of $\psi_l$ is linear in $l$ and therefore $d\leq 2^{k^y}$ for some $y$. Altogether this means that the size of $\psi_l$:
\[
  |\psi_l| \leq m^d \leq (2^{2^{k^x}})^{2^{k^y}}
\]
is at most doubly exponential in $k$. It remains to prove Lemma~\ref{lem:construction}.

\begin{proof}
Let $\A=(Q,\Sigma,q_I,\delta,\rank)$ be the NPTA equivalent to $\phi$
with exponentially many states, which exists due to Theorem~\ref{thm:munpta}. 
For each $n\in\NN$ and $q\in Q$ we construct $\psi_{n,q}\in\ML^n$ of branching $2^{2^{|Q|}}$ such that:
\begin{align}
  \M\models\psi_{n,q} \text{\hspace*{0.5cm} $\iff$ \hspace*{0.5cm} there exists $\N\models\A[q_I\mapsfrom q]$ with $\M\bis^n\N$}\label{eq:universal consequence}
\end{align}
for every structure $\M$. Then, $\psi_{n,q_I}$ is our desired $n$-uniform consequence $\psi_n$ of $\phi$.

We proceed by induction on $n\in\NN$. For the base case we put:
\[
  \psi_{0,q} = \bigvee\{c\in\alphabet\ |\ \text{there is $\N\models\A[q_I\mapsfrom q]$ with $\N\models c$}\}
\]
which clearly satisfies the induction goal~\eqref{eq:universal consequence}. For the induction step define:
\[
  \psi_{n+1,q} = \bigvee_{c\in\alphabet} \bigvee_{S\in\delta(q,c)}c \wedge \nabla \{\psi_{n,p}\ |\ p\in S\}.
\]
Fix $\M$, denote the color of the root by $c$ and the set of all children of the root by $M_0$.
If $\M$ satisfies $\psi_{n+1,q}$ then there is $\{p_1,...,p_l\}=S\in\delta(q,c)$ such that nabla of $\Phi=\{\psi_{n,p}\ |\ p\in S\}$ is satisfied in the root. By invariance under bisimulation we may assume that the root of $\M$ has sufficiently many children to find a separate witness for each $\psi_{n,p}\in\Phi$. That is, we assume a surjective assignment $h:M_0\to\Phi$ that maps every $\point\in M_0$ to some formula $\psi_{n,p}$ true in $\point$. By induction hypothesis, for each $\point\in M_0$ with $h(\point)=\psi_{n,p}$ there is a model $\N_p\models\A[q_I\mapsfrom p]$ $n$-bisimilar to the subtree of $\M$ rooted in $\point$. Define $\N$ as follows: first take the disjoint union $\{\point\}\sqcup\bigsqcup\{\N_p\ |\ p\in S\}$ of all the $\N_p$'s and a fresh point $v$ of color $c$; then for every $\N_p$ add an edge from $\point$ to the root of $\N_p$ and set $\point$ as the new root. It is easy to see that $\N\models\A[q_I\mapsfrom q]$ and $\M\bis^{n+1}\N$, as desired.

Conversely, assume $\M\bis^{n+1}\N$ and $\N\models\A[q_I\mapsfrom q]$ witnessed by an $(n+1)$-bisimulation $Z$ and a run $\rho:N\to Q$. Denote the children of the root of $\N$ by $N_0$. 
Since $\rho$ is a run, the set $S=\rho[N_0]$ of states assigned to $N_0$ belongs to $\delta(q,c)$. Every $\point\in M_0$ is $n$-bisimilar to some $\altpoint\in N_0$ and hence by the induction hypothesis satisfies $\psi_{n,p}$ for $p=\rho(\altpoint)\in S$. Symmetrically, for every $p\in S$ there is $\altpoint\in N_0$ accepted by $\A[q_I\mapsfrom p]$. Since that $\altpoint$ is $n$-bisimilar to some $\point\in M_0$, by induction hypothesis $\point$ satisfies $\psi_{n,p}$. It follows that the root of $\M$ satisfies $\nabla\{\psi_{n,p}\ |\ p\in S\}$ and therefore also $\psi_{n+1,q}$.
\end{proof}

Let us remark that Lemma~\ref{lem:construction} can be easily adapted to deal with vocabulary restrictions. That is, given $P\subset\Propositions$ we could construct $\psi_n^{P}$ similar to $\psi_n$ but only using atomic propositions from $P$ and only entailing $\ML^n$-consequences of $\phi$ whose vocabulary is contained in $P$. To that end, it suffices to project-out atomic propositions not in $P$ from the automaton $\A$ and only then proceed with our construction. 
The automaton $\A^P$ obtained by such a projection accepts a model $\M$ iff $\M$ is $P$-bisimilar to some $\N\models\A$. It follows that for every $\theta\in\ML^{n,P}$:
\begin{align*}
  \psi_n^P\models \theta \hspace*{0.5cm} \iff \hspace*{0.5cm} \A^P\models \theta \hspace*{0.5cm} \iff \hspace*{0.5cm} \A\models\theta
\end{align*}
where the first equivalence uses $\theta\in\ML^n$ and the latter $\theta\in\ML^P$.
Such \emph{$(P,n)$-uniform consequence} $\psi_n^{P}$ of $\phi$ can then be taken as a Craig modal separator, in the same way as $\psi_n$ serves as a modal separator.

\subsection{Lower bounds}
For the lower bounds, we show that over arbitrary structures (in fact,
already binary trees) $\muML$ is doubly exponentially more succinct
than $\ML$. The example is essentially taken from \cite[Section
3.1]{French13}. There the authors use game-theoretic tools which are
later applied to more complicated cases. Since we are only interested
in this example, we provide a straightforward self-contained
argument.
\begin{proposition}\label{prop:succinctnes all models}
There is a sequence $(\phi_n)_{n\in\NN}$ of $\muML$-formulae of size
polynomial in $n$ such that each $\phi_n$ is equivalent to a $\ML$-formula but every $\psi\in\ML$ equivalent to $\phi_n$ has size at least $2^{2^n}$.
\end{proposition}
\begin{proof}
We only give a sketch, the details are found in
Appendix~\ref{app:lower bounds}. We assume two different actions $\letterA$ and $\letterB$. For each $n\in\NN$ consider the property:
\begin{itemize}
  \item[$B_n$:] ``No path (over all actions) longer than $2^n$ starts in the root.''
\end{itemize}
This can be enforced by encoding an $n$-bit binary counter into the
structure of the model, and requiring that on every path the counter
values are strictly increasing. Let $C_n$ be this (technically
stronger) property expressing the behavior of the encoded counter.
Assuming that the encoding is reasonably efficient, $C_n$ can be
easily expressed by a $\muML$-formula $\phi_n$ of size polynomial in $n$ (in fact, a weak fragment of $\PDL$ is already sufficient). Since the lengths of paths are bounded, $C_n$ can be also expressed in $\ML$.

However, every $\psi\in\ML$ equivalent to
$\phi_n$ has size at least $2^{2^n}$. The reason is that for every
sequence of actions $\letterA$ and $\letterB$ of length $2^n$, the
syntax tree of $\psi$ must contain a descending sequence of subformulae of length $2^n$ such that the $i$-th subformula begins with a modal operator corresponding to the $i$-th action. This allows to embed a binary tree of height $2^n$ into the syntax tree of $\psi$.
\end{proof}
Note that the presence of two different actions $\letterA$ and
$\letterB$ is essential for the argument. We conjecture that $\muML$
is doubly exponentially more succinct than $\ML$ already in the
monomodal setting. 
% Let $\phi_n^\circ$ use a binary counter to enforce
% an exponential bound on the length of paths, the same way as $\phi_n$
% except that now in presence of only one modality. 
% We expect following sequence of formulae to witness that:
Consider the following Property $P_n$, parameterized by $n\in \NN$:
\begin{itemize}
  \item[$P_n$:] ``$C_n$ and there exists a maximal path on which the number of points satisfying $\atProp$ is even.''
\end{itemize}
where $C_n$ is the same as in Proposition~\ref{prop:succinctnes all models}.
It is not difficult to come up with small, that is, of size polynomial
in $n$, $\muML$-formulae
$\varphi'_n$ expressing $P_n$. Unfortunately, proving that no small $\ML$
formula can be equivalent to $\phi_n'$ seems difficult. For instance,
consider models where every non-leaf point has a child satisfying $\atProp$ and
a child satisfying $\neg\atProp$. Then a trick similar to the famous
example of Potthoff (showing, roughly, that the
language of all binary trees of even depth is first-order
definable)~\cite[Example 1]{Potthoff95} can be exploited to get a modal formula equivalent to $\phi_n'$ (over such models), but of size only single exponential in $n$. Moreover, the results in the next Section~\ref{sec:separators - words} show that looking at words only is not sufficient either.

\section{Optimal Separators: Word Case}\label{sec:separators - words}

In this section we show that optimal modal separators (over words)
can be computed exponentially faster and are exponentially smaller
compared to the case with arbitrary models. 
\begin{theorem} \label{thm:word-construction} 
  If $\varphi$ and $\varphi'$ are modally separable over
  words, then a separator of size exponential in
  $|\varphi|+|\varphi'|$ exists and can be computed in exponential
  time.
\end{theorem}
As with arbitrary models, Proposition~\ref{prop:bound on nesting}
gives an upper bound on the modal depth of a separator and so it suffices to construct $n$-uniform consequences of $\phi$ of small size.

We illustrate the idea first. Consider the
classes \textsf{EVEN}$_n$ and \textsf{ODD}$_n$, $n\in\NN$ of all word structures of length $n$ in which proposition $a$ is satisfied in an
even and odd, respectively, number of points.  Constructing modal
formulae $\varphi_n$ and $\varphi'_n$ defining \textsf{EVEN}$_n$ and
\textsf{ODD}$_n$ in the following, naive way leads to exponential
formulae since $\varphi_{i+1}$ contains both and $\varphi_i$ and
$\varphi'_i$: 
\begin{align*}
  \varphi_0 & = \neg a \wedge \Box\bot &
  \varphi_{i+1} & = \Diamond\top \wedge \big( (a\wedge \varphi_i')\vee
  (\neg a\wedge \varphi_i)\big) \\ 
  \varphi_0' & = a \wedge \Box\bot &
  \varphi_{i+1}' & = \Diamond\top \wedge \big( (a\wedge \varphi_i)\vee
  (\neg a\wedge \varphi_i')\big)
\end{align*}
This exponential blow-up can be avoided, however, using
``divide-and-conquer'' as follows:
\begin{align*}
  \varphi_{2n} = \big(\varphi_n \wedge \Diamond^n\varphi_n\big) \vee
  \big(\varphi_n'\wedge \Diamond^n\varphi_n'\big)\\
  \varphi_{2n}' = \big(\varphi_n \wedge \Diamond^n\varphi_n'\big) \vee
  \big(\varphi_n'\wedge \Diamond^n\varphi_n\big)
\end{align*}
Although several copies of formulae of smaller index are used as well,
but since the index is halved, we end up with formulae of roughly
quadradic size. The proof of the following analogue of
Lemma~\ref{lem:construction} relies on this idea.
\begin{lemma}\label{lem:word-construction}
  For every $n\in\NN$ and every NPWA $\Amc$ with states $Q$, one can
  construct a formula $\psi_n\in\ML^n$ which is $\A$'s $n$-uniform consequence over words and has size polynomial in $n$ and $|Q|$. The construction requires polynomial time.
% 
%   for every word structure $\M$: % \begin{align} % \M\models\psi_n
%   \text{\hspace*{0.5cm} $\iff$ \hspace*{0.5cm} there exists
%   $\N\models \Amc$ with $\N\leftrightarroweq^n \M$ }
%   \label{eq:word-universal-consequence} % \end{align} 
  %
\end{lemma}
To see that Lemma~\ref{lem:word-construction} implies
Theorem~\ref{thm:word-construction}, let $\varphi$ and $\varphi'$ admit
a modal separator over words.  Let
$\Amc$ be an NPWA that is equivalent to $\varphi$. 
By Theorem~\ref{thm:munpta}, \Amc has
exponentially many states and can be computed in exponential time.
Proposition~\ref{prop:bound on nesting} implies that there is a
modal separator of modal depth $l$ at most exponential in $k=|\varphi|+|\varphi'|$.
As with arbitrary models, $\Amc$'s $l$-uniform consequence $\psi_l$ from
Lemma~\ref{lem:word-construction}
is the sought separator. We now prove the lemma.
\begin{proof}
Let $\Amc=(Q,\Sigma,\delta,q_I,\rank)$ be an NPWA. The main idea is to
construct, for every $p,q\in
Q$ and $m\in\NN$, a formula $\psi_{p,q}^m$ such that for every input word $\Mmc$:
\[
  \Mmc\models\psi_{p,q}^m \iff \text{there is a run from $p$ to $q$
  over the $m$-prefix of $\Mmc$},
\]
The key step is the recursive splitting similar to the definitions of \textsf{EVEN}$_n$ and
\textsf{ODD}$_n$ above. Intuitively, $\psi_{p,q}^{2m}$ is the disjunction
over all $s\in Q$ of the conditions ``there is a run from $p$ in the
initial position to $s$ in position $m$, and a run from $s$ in
position $m$ to $q$ in position $2m$.'' The latter conditions are
recursively expressed using $\psi_{p,s}^m$ and $\psi_{s,q}^m$. The
constructed formulas $\psi_{q_I,q}^m$, $m\leq n$ are then used to describe all
possible $n$-prefixes of models of $\Amc$. The details of the construction can be found in Appendix~\ref{app:optimal separators - words}.
\end{proof}

We conclude the section with the comment that
Theorem~\ref{thm:word-construction} is optimal in the sense that there
are modally separable formulae which require a large separator. We
actually show the following stronger statement implying that,
over words, $\muML$ is exponentially more succinct than $\ML$.
\begin{proposition} \label{prop:words-lower}
  There is a sequence of $\muML$-formulae $(\varphi_n)_{n\in\NN}$ 
  of size polynomial in $n$ such that each $\varphi_n$ is equivalent
  to a $\ML$-formula but every $\psi\in\ML$ equivalent to $\varphi_n$ has size at least
  $2^n$.
\end{proposition}
The proof is entirely standard. The main idea is that, already in
$\PDL$ one can stipulate (with a small formula) a finite word of exponential length. Clearly,
any $\ML$-formula expressing this requires exponential size. The
only difficulty is doing it with a fixed signature: instead of
encoding $i$-bit counters using $i$ propositions, we use just two
propositions and encode numbers in $i$ consecutive points.

\section{Conclusion and Open Problems}

We have studied the problem of deciding separability of
$\muML$-formulae by fixpoint free formulae from $\ML$, and computing
separators if they exist. Our results cover several interesting
classes of models such as trees, finite trees, and words. Due to the great
expressivity of $\muML$ the results remain valid in the presence of ontologies.

A notably missing case is the class of trees of fixed outdegree $d$ independent from formulae. This is surprisingly different from the classes we studied. The key difficulty here lies in the fact that the implication \ref{it:bound identical}~$\Rightarrow$~\ref{it:bound bisimilar} from Proposition~\ref{prop:separability reformulation} is not true over such trees.

% 
% In the future, we plan to modify our construction of the separators to be able to
% return also Craig separators.  
An intriguing challenge left for future study is to look at
extensions of $\muML$ and/or $\ML$. Natural extensions are inverse
modalities, the universal modality, graded modalities, and constants
(corresponding to inverse roles, the universal role, counting
quantifiers, and nominals in DL speech). We expect the adaptation to
inverse modalities to be only minor. Also graded modalities look
innocent if they are allowed both in the larger logic and in the
separator logic. If we only extend $\muML$ with graded modalities and
ask for separators in $\ML$ (without graded modailites), we would have to combine our techniques with
the ones from~\cite{separatingcounting}, which is potentially challenging. 
We expect universal modality and/or constants to
pose more technical difficulties as well.  Intuitively, adding a universal
modality or constants leads to the loss of the strong locality
underlying Proposition~\ref{prop:separability reformulation}. 

\newpage
\bibliographystyle{plain}
\bibliography{modsep}

\newpage

\appendix
Throughout the appendix we use the term \emph{tallness}. The tallness of a tree is the distance from the root to the \emph{closest} leaf (or $\infty$ if it the tree has no leafs).

\section{Separability Coincides with Craig Separability}\label{app:craig}
We prove Theorem~\ref{thm:craig} which says that $\phi,\phi'$ are modally separable iff they are Craig modally separable.
\begin{proof}
  Clearly, any Craig modal separator is also a modal
  separapator.
 
  Conversely, suppose there is a modal separator $\psi$ of
  $\varphi,\varphi'$, and let $n$ be its modal depth. We use the
  $n$-uniform consequences as defined in Definition~\ref{def:n-uniform
  consequence}. Let $\theta$ be the $n$-uniform consequence of
  $\varphi$ and $\theta'$ the $n$-uniform consequence of $\varphi'$
  (both exist, see the discussion after Definition~\ref{def:n-uniform
  consequence}). Note that $\neg\theta'\models \neg\varphi'$. By
  definition of $n$-uniform consequence, we have $\theta\models\psi$
  and $\psi\models \neg\theta'$, and thus $\theta\models\theta'$.
  Since $\ML$ enjoys Craig interpolation, there is an interpolant
  $\psi'$ for $\theta,\theta'$ which is also a Craig modal separator
  of~$\varphi,\varphi'$.
\end{proof} 
\section{Model-theoretic Characterization}\label{app:model theory}
\noindent
We prove Proposition~\ref{prop:separability reformulation}.\\

\noindent
The implication~\ref{it:separability}~$\Rightarrow$~\ref{it:bound id + bounded} is straightforward. Separator $\psi\in\ML^n$ cannot distinguish models identical up to depth $n$. If $\M\models\phi$ then $\M\models\psi$ and so $\M'\models\psi$ which in turn implies $\M'\models\neg\phi'$.\\

\noindent
Our argument for~\ref{it:bound id + bounded}$\implies$\ref{it:bound
identical} uses a classical characterization of the semantics of
$\muML$ formulae in terms of parity games. Since we do not want to
introduce games we only sketch the (easy) construction; the necessary
definitions can be found in~\cite{Venema20}.

We prove the implication by contrapositive. Assume $\M\models\phi$ and $\M'\models\phi'$ identical up to depth $n$. The fact that $\M\models\phi$ is equivalent to existence of a winning strategy in an appropriately defined parity game $\game_\phi$. Positions of the game are pairs: $(\point,\theta)$ with $\point\in M$ and $\theta$ subformula of $\phi$ and a move from $(\point,\theta)$ to $(\point',\theta')$ is only allowed if $\point\arrowActionLabel{}\point'$. A similar game $\game_{\phi'}$ captures the semantics of $\phi'$.

Take \emph{positional} winning strategies $\sigma$ and $\sigma'$ for $\eve$ in the games $\game_\phi$ and $\game_{\phi'}$. We trim $\M$: only keep the points that belong to a position chosen by $\eve$ in either a $\sigma$- or $\sigma'$-play, and remove all the other ones. After that, also remove points that become inaccessible from the root so that the resulting structure is a tree. The tree $\N$ obtained this way still satisfies $\phi$ because the strategy $\sigma$ remains winning. For every $\point\in M$ both games together only have $k$ positions containing $\point$. Thus, by positionality, among the children of $\point$ only at most $k$ many belong to a position chosen by $\sigma$ or $\sigma'$. Since only such points belong to $N$, it follows that $\N$ has branching at most $s$. We trim $\M'$ to $\N'\models\phi'$ in the same way. By construction $\N$ and $\N'$ are identical up to depth $n$ which violates \ref{it:bound id + bounded}.
\\

\noindent
To prove~\ref{it:bound identical}~$\Rightarrow$~\ref{it:bound bisimilar} assume towards contradiction that there are models $\M\models\phi$ and $\M'\models\phi'$ linked by an $n$-bisimulation $Z$. We construct trees $\M_Z\models\phi$ and $\M_Z'\models\phi'$ identical up to depth $n$, therefore reaching a contradiction with~\ref{it:bound identical}.

Without loss of generality we assume that $\M$ and $\M'$ are trees, otherwise they can be unravelled. The tree $\M_Z$ is as follows. It has universe:
\[
  M_Z=M\cup Z
\]
and the pair consisting of the roots of $\M$ and $\M'$ (which by assumption belongs to $Z$) is taken as the new root. Edges between pairs from $Z$ are defined pointwise (that is: $(\point,\point')\arrowActionLabel{}(\altpoint,\altpoint)$ iff $\point\arrowActionLabel{}\altpoint$ and $\point'\arrowActionLabel{}\altpoint'$). Edges between points in $M$ are taken from $\M$. On top of that, we add $(\point,\point')\arrowActionLabel{}\altpoint$ iff in $\M$ the point $\point$ is at depth precisely $n$ and $\point\arrowActionLabel{}\altpoint$. The colors are inherited from $\M$ on $M$ and from whichever coordinate on $Z$ (points linked by $Z$ always have the same color).

Consider the function $f:M_Z\to M$ defined as the left projection on $Z$ and identity on $M$:
\begin{align*}
  f(\point,\point')& =\point\\
  f(\point)& =\point.
\end{align*}
The graph of $f$ is a bisimulation between $\M_Z$ and $\M$. In particular, $\M_Z$ satisfies $\phi$. We define $\M_Z'$ satisfying $\phi'$ symmetrically. Then $\M_Z$ and $\M_Z'$ satisfy $\phi$ and $\phi'$ and are identical up to depth $n$ (the $n$-prefixes of both $\M_Z$ and $\M_Z'$ are equal $Z$). Technically, $\M_Z$ and $\M_Z'$ are directed acyclic graphs but not necessarily trees. However, they can be turned into trees by removing inaccessible points and unravelling. Both these operations preserve satisfaction of $\phi$ and $\phi'$ and identity of $n$-prefixes.
\\

\noindent
For the implication~\ref{it:bound bisimilar}~$\Rightarrow$~\ref{it:separability} one can define $\psi$ as (any) $n$-uniform consequence of $\phi$. An explicit instance of such $n$-uniform consequence is the disjunction of all $\ML^n$-types consistent with $\phi$ (see Definition~\ref{def:n-uniform consequence} of $n$-uniform consequences and the following discussion).

\section{Separability over All Models is in \ExpTime}\label{app:deciding and bounds - all models}
We prove the upper bounds in Theorem~\ref{thm:models-complexity}. By Proposition~\ref{prop:separability reformulation} deciding modal separability boils down to checking if there is $n\in\NN$ for which~\ref{it:bound id + bounded} holds. Using well-known properties of $\ML$ and $\muML$ (namely, closure under relativization) we may reduce the problem to the special case when models under consideration are \emph{full} $k$-ary trees, meaning that every point has precisely $k$ children. Under this assumption the graphs underlying models are all isomorphic to the (unlabelled) full $k$-ary tree $\KK=(K,\arrowActionLabel{})$. Hence, we identify models with valuations $\val:K\to\alphabet$. Let us call a finite prefix $X\subset K$ \emph{sufficient for separation} if for every $\val,\val':K\to\alphabet$ identical on $X$: $\val\models\phi$ implies $\val'\models\neg\phi'$. Denote the set of such prefixes:
\[
  \sepLang = \{X\subset K\ |\ \text{$X$ is a finite prefix of $K$ sufficient for separation}\}.
\]
It follows directly from the definition that for every $n\in\NN$:
\begin{align}
  \text{The $n$-prefix $K_{|_{n}}$ of $K$ is in $\sepLang$ $\iff$ Condition~\ref{it:bound id + bounded} is true for $n$.}\label{eq:lang vs separability}
\end{align}
The set $\sepLang$ can be viewed as a language of finite trees. This language is closed under taking finite supermodels, so it contains
$K_{|_{n}}$ for some $n$ iff it is nonempty. It follows that
separability of $\phi$ and $\phi'$ is equivalent to nonemptiness of
$\sepLang$.

We construct an automaton $\B$ that accepts finite trees \emph{not} belonging to $\sepLang$. Take automata $\A$ and $\A'$ equivalent to $\phi$ and $\phi'$ of size exponential in $k$. The idea is that given a finite tree $X\subset K$ the automaton $\B$ guesses $\val:X\to\alphabet$, $\rho:X\to Q$ and $\rho':X\to Q'$ that can be extended to a valuation $\val_+:K\to\alphabet$ and accepting runs $\rho_+:K\to Q$ of $\A$ and $\rho_+':K\to Q'$ of $\A'$ consistent with $\val_+$. $\B$ has the set $Q^\B=Q\times Q'$ as states and $(q_I,q_I')$ is the initial one. Transition function is defined in two steps. First take:
\begin{gather*}
  R\in\delta_0^\B(q,q')\\
  \iff \\
  \text{There is $c\in\alphabet$, such that the left projection of $R$}\\
  \text{belongs to $\delta(q,c)$ and the right one to $\delta'(q',c)$}.
\end{gather*}
This describes guessing a coloring on the full $\KK$ and runs of both $\A$ and $\A'$ consistent with that coloring. To handle points with less than $k$ children we put:
\begin{gather*}
  R\in\delta^\B(q,q')\\
  \iff\\
  \text{$R$ can be obtained from some $R'\in\delta_0^\B(q,q')$}\\
  \text{by removing some consistent pairs.}
\end{gather*}
Here consistency of a pair of states $(p,p')$ means that there exists a model accepted by both $\A[q_I\mapsfrom p]$ and $\A'[q_I'\mapsfrom p']$. The (trivial) rank function of $\B$ assigns a bad rank to every state so that only finite trees are accepted.

To finish the proof let us first prove Proposition~\ref{prop:bound on nesting}.
\begin{proof}
Put $l=|Q|\times|Q'|+1$. The proposition follows directly from \eqref{eq:lang vs separability} and:
\begin{align}
\text{$\sepLang$ is nonempty $\iff$ it contains the $l$-prefix $K_{|_{l}}$ of $K$.}\label{claim:bound SEP}  
\end{align}
Only the left-to-right implication is nontrivial and we prove it by contrapositive. If $K_{|_{l}}$ does not belong to $\sepLang$ then this is witnessed by a run $\rho$ of $\B$ on $K_{|_{l}}$. Since every leaf $\point$ of $K_{|_{l}}$ is at depth greater then $|Q^\B|$, on the path from the root to $\point$ some state must repeat. Hence, the run can be pumped to finite prefixes of $K$ of arbitrarily big tallness. It follows that $\sepLang$ is empty. This completes the proof of Proposition~\ref{prop:bound on nesting}.
\end{proof}

The language of $\B$ is closed under taking submodels. Hence, the right hand side of~\eqref{claim:bound SEP} is equivalent to $\B$ \emph{not} accepting any tree of tallness at least $l$ ($l$ is the exponential bound given by Proposition~\ref{prop:bound on nesting}). This can be easily checked in time polynomial in the size of $\B$.

\section{Separability over Words is in \PSpace}\label{app:deciding and bounds - words}
We prove the upper bounds in Theorem~\ref{thm:word-complexity}. We proceed with the same argument as in the case with general models, with two major simplifications. First, over words bisimilarity and identity coincide. Thus the equivalence of items~\ref{it:bound bisimilar},~\ref{it:bound identical} and~\ref{it:bound id + bounded} of Proposition~\ref{prop:separability reformulation} becomes trivial in the word case. Second, bounded tallness (which over words is the same as finiteness) of the language of the appropriate word automaton $\B$ can be checked more easily. Instead of writing down the automaton we construct it on the fly, and only remember a single state at every moment. This yields a \PSpace (and not \ExpTime) decision procedure.
% The rest of the construction is the same. We define an appropriate language $\sepLang$ of \emph{finite words} over unary alphabet.
% % Thus in place of Proposition~\ref{prop:separability reformulation} we have a rather immediate equivalence:
% % \begin{gather*}
% %   \text{There is $\psi\in\ML$ of modal depth $n$ separating $\phi$ and $\phi'$ over words.}\\
% %   \iff\\
% %   \text{Every finite word of length $n$ that is a prefix of a word}\\
% %   \text{satisfying $\phi$ is \emph{not} a prefix of a word satisfying $\phi'$.}
% % \end{gather*}
% % which holds for all $n\in\NN$.
% Nonemptiness of this language is equivalent to separability and its complement is recognized be exponential-sized $\B$. As in the general case we get:

\section{Lower Bounds over All Models}\label{app:lower bounds}
We construct formulae of $\muML$ for which there exist equivalent formulae in $\ML$ but only doubly exponentially larger.

Consider the nonstandard modal operator $\diamond{\exists^r}$ which means ``there exists a point reachable from the root and satisfying''. It is straightforward to encode such modality (and its dual $\boxmodal{\exists^r}$ meaning ``for every reachable point'') by a $\muML$ formula of constant size. It therefore suffices to define the sequence $(\phi_n)_{n\in\NN}$ in the extension $\ML+\exists^r$ of $\ML$ with such operators.

We construct $\phi_n$. For convenience assume $n+1$ atomic propositions $\atProp_0,...,\atProp_n$ (otherwise we encode these without a significant blowup in the size of $\phi_n$). The propositions can be seen as an encoding of an $(n+1)$-bit binary counter. Consider the following property $C_n$ of models: the color of the root encodes counter value $0$ and for every (reachable) point $\point$ if the color of $\point$ encodes number $i$ then either $i=2^n$ and $\point$ has no children or all $\point$'s children have a color that encodes $i+1$. The property $C_n$ can be easily expressed by a $\ML+\exists^r$ formula $\phi_n$ of size polynomial in $n$. Such $\phi_n$ is equivalent to a modal formula. This follows immediately from Proposition~\ref{prop:separability reformulation} since property $C_n$ implies that there are no paths longer than $2^n$.

However, no $\ML$ formula smaller than $2^{2^n}$ can be equivalent to $\phi_n$. To see this assume $\psi\in\ML$ equivalent to $\phi_n$. For every sequence of colors $\letterA_1,...\letterA_{2^n}$ of length $2^n$, $\psi$ must contain a sequence $\xi_1,...,\xi_{2^n}$ of subformulae such that:
\begin{itemize}
  \item $\xi_i$ is a strict subformula of $\xi_j$ whenever $i>j$;
  \item each $\xi_i$ begins with $\diamond{\letterA_i}$ or $\boxmodal{\letterA_i}$;
  \item the least strict superformula of $\xi_i$ beginning with a modal operator is $\xi_{i+1}$.
\end{itemize}
If $\psi$ did not contain such sequence $\xi_1,...,\xi_{2^n}$, it would be indifferent to what happens in points only reachable via paths labelled with $\letterA_1,...\letterA_{2^n}$. It follows that there is an embedding of the $2^n$-prefix of the full binary tree to the syntax tree of $\psi$ and so it has size at least $2^{2^n}$.

Note that the use of multiple atomic propositions in the above construction is only for convenience. With any reasonable encoding the proof could be adapted even to the setting with $\Propositions=\emptyset$.

\section{Optimal Separators over Words}\label{app:optimal separators - words}
We prove Lemma~\ref{lem:word-construction}. Recall that $\Amc=(Q,\Sigma,\delta,q_I,\rank)$ is the automaton under consideration. Since we are working over words, $\delta(q,c)$
contains only singleton sets and possibly the emtpy set. The latter
case, $\emptyset\in \delta(q,c)$, is of particular interest because
this means that the automaton in state $q$ reading color $c$ can
``accept'' even if it has not finished reading the input; in
particular, the automaton can accept finite words as well. Denote with
$\textit{Acc}_q$ the set of colors $c$ with $\emptyset\in
\delta(q,c)$. Further denote with $\textit{Cont}_q$ the set of all
$c$ such that $\Amc[q_I\mapsfrom q]$ accepts a word starting with color $c$.

% 
% over $\omega$-words that accepts
% precisely the encodings $\pi_\M$ of word structures $\M$ that are
% models of $\varphi$. 
% 
% We refrain from giving more details on the
% (standard) encoding of word structures as $\omega$-words, and only
% mention that the transition relation $\Delta$ contains all triples
% $(p,c,q)$ such that \Amc can change from state $q$ to state $p$ when
% reading input letter $c\subseteq\Propositions$. It is well-known that such an $\Amc$ always exists
% and can be computed in exponential time~\cite{X,Y}. Assume without
% loss of generality that each state of $\Amc$ is terminating.

We first construct $\ML$ formulae $\psi^m_{pq}$, for $m\in\NN,
p,q\in Q$ such that for every word $\M$: 
\begin{itemize}

  \item[$(\ast)$] $\M\models \psi^m_{pq}$ iff there is a run of \Amc
    from $p$ to $q$ on the $m$-prefix of $\M$.

\end{itemize}
 The definition is by induction on $m$:
\begin{align*}
  \psi^0_{pq} & = \text{ if $p\neq q$ then $\bot$ else $\top$} \\
  \psi^1_{pq} & = \bigvee\{c\ \mid\ c\in\alphabet, \{q\}\in \delta(p,c)\} \\
  \psi^{m}_{pq} & = \bigvee_{q'\in Q} \psi^{\lfloor m/2\rfloor }_{pq'}
  \wedge\psi^{\lceil m/2\rceil }_{q'q}
\end{align*}
It is routine to verify that $\psi^m_{pq}$ satisfies~$(\ast)$ and is
of size $|\psi^m_{pq}|\in O(|Q|\cdot m^2)$.

% \medskip\noindent\textit{Claim.} $|\psi^m_{pq}|\in O(|Q|\cdot m^2)$.

% \medskip 
We finish the construction by setting:
\[\psi_n = \bigvee_{q\in Q} \left(\psi^n_{q_Iq}\wedge 
  \Box^n\bigvee_{c\in \textit{Cont}_q} c\right)\vee
  \bigvee_{m\leq n}\bigvee_{q\in Q} \left(\psi^m_{q_Iq}\wedge \Box^{m+1}\bot \wedge
\Box^{m}\bigvee_{c\in \textit{Acc}_q}c\right).\]
It is readily checked that $\psi_n$ satisfies the required size bounds.
To verify that $\psi_n\models \theta$ for every $\theta\in \ML^n$ with
$\Amc\models \theta$, we show the following equivalence for all
words $\Mmc$:
\begin{align}
  \M\models\psi_n \text{\hspace*{0.5cm} $\iff$ \hspace*{0.5cm} there exists
  $\N\models \Amc$ with $\N\leftrightarroweq^n \M$.}
  \label{eq:word-universal-consequence} 
\end{align} 

For $\Rightarrow$, let $\M$ be a word structure with $\M\models
\psi_n$. If $\M\models \psi^n_{q_Iq}\wedge
  \Box^n\bigvee_{c\in\textit{Cont}_q}c$ for some $q$, then by~$(\ast)$,
there is a run of $\Amc$ from the initial state $q_0$ to some state
$q\in Q$ when reading the $n$-prefix of $\M$, and the last color in
the prefix is $c$. Since $c\in \textit{Cont}_q$, we can extend the $n$-prefix of $\M$ to
a word $\N$ accepted by \Amc. If $\M\models
\psi^m_{q_Iq}\wedge\Box^{m+1}\bot\wedge\bigvee_{c\in
\textit{Acc}_q}c$, for some $m\leq n$ and $q\in Q$, then $\M$ is a finite
word of depth $m$ that is accepted by the automaton. We can take
$\N=\M$ in this case. 

For $\Leftarrow$, let $\M$ be a word such that there is some
$\N\models\varphi$ with $\N\leftrightarroweq^n \M$. The former
condition implies that $\N\models \Amc$ and thus there is an accepting
run $\rho$ of $\Amc$ on
$\N$, and the latter
implies that $\N$ and $\M$ coincide on their
$n$-prefixes. We distinguish cases. If the depth of $\N$ is greater than
$n$, then the $n$-prefix of $\rho$ ending in state $q$ witnesses 
$\M\models \psi^n_{q_Iq}\wedge \Box^n\bigvee_{c\in\textit{Cont}_q}c$.
Otherwise, the depth of $\N$ is $m\leq n$ and the run $\rho$ ending in
$q$ witnesses that $\M\models
\psi^m_{q_Iq}\wedge\Box^{m+1}\bot\wedge\bigvee_{c\in
\textit{Acc}_q}c$.

\end{document}